\def\trversion{}

\documentclass[runningheads]{llncs}

\newcommand{\iftr}[2]{\ifx\trversion\undefined{#2}\else{#1}\fi}
\newcommand{\ifshort}[2]{\ifx\trversion\undefined{#1}\else{#2}\fi}

\usepackage{proof}
\usepackage{listings}
\usepackage{letltxmacro}
\usepackage{rtyps}
\usepackage{amsmath}
\usepackage{amssymb}
\usepackage{adjustbox}
\usepackage{mathpartir}
\usepackage{float}
\usepackage{marvosym}
\usepackage{hyperref}
\DeclareRobustCommand{\stirling}{\genfrac\{\}{0pt}{}}
\newcommand{\qprove}[2]{\rtyps{#1}{#2}}

\title{Exponential Automatic Amortized Resource Analysis%
  \thanks{%
    This article is based on research supported by DARPA under AA
    Contract FA8750-18-C-0092 and by the National Science Foundation
    under SaTC Award 1801369, SHF Award 1812876, and CAREER Award
    1845514.
    Any opinions, findings, and conclusions contained in this document
    are those of the authors and do not necessarily reflect the views of
    the sponsoring organizations.
}}

\titlerunning{Exponential AARA}

\author{David M. Kahn [\Letter] \and
Jan Hoffmann}

\institute{Carnegie Mellon University, Pittsburgh PA, USA
\\
\email{davidkah@cs.cmu.edu}
\url{cs.cmu.edu/~davidkah}
\\
\email{jhoffmann@cmu.edu}
\url{cs.cmu.edu/~janh} }

\authorrunning{D.\,M. Kahn and J. Hoffmann}

\begin{document}

\maketitle
\begin{abstract}
  Automatic amortized resource analysis (AARA) is a type-based
  technique for inferring concrete (non-asymptotic) bounds on a
  program's resource usage. Existing work on AARA has focused on
  bounds that are polynomial in the sizes of the inputs.
  This paper presents and extension of AARA to exponential bounds that
  preserves the benefits of the technique, such as compositionality
  and efficient type inference based on linear constraint solving.
  A key idea is the use of the Stirling numbers of the second kind as
  the basis of potential functions, which play the same role as the
  binomial coefficients in polynomial AARA.
  To formalize the similarities with the existing analyses, the paper
  presents a general methodology for AARA that is instantiated to the
  polynomial version, the exponential version, and a combined system
  with potential functions that are formed by products of Stirling
  numbers and binomial coefficients.
  The soundness of exponential AARA is proved with respect to an
  operational cost semantics and the analysis of representative
  example programs demonstrates the effectiveness of the new analysis.
  \vspace{-1ex}
\end{abstract}

\section{Introduction}
\label{sec:Intro}

``Time is money'' is a phrase that also applies to executing software,
most directly in domains such as on-demand cloud computing and smart
contracts where execution comes with a explicit price tag.
In such domains, there is an increasing interest in formally analyzing
and certifying the precise resource usage of programs.  However, the cost
of formally verifying properties by hand is an obstacle to even
getting projects off the ground. For this reason, it would be
desirable if such resource analyses could be performed mostly
automatically, with reduced burden on the programmer.

Techniques and tools for automatic and semi-automatic resource
analysis have been extensively studied. The applied methods range from
deriving and analyzing recurrence
relations~\cite{Wegbreit75,Albert07,FloresH14,AlbertFR15,DannerLR15,KincaidBBR2017,ChatterjeeFG17,Kincaid18},
to abstract interpretation and static
analysis~\cite{GulwaniMC09,BlancHHK10,SinnZV14,MadhavanKK17}, to type
systems~\cite{LagoG11,CicekBGGH16,OOPSLA:WWC17}, to proof assistants
and program
logics~\cite{Atkey10,CarbonneauxHZ15,CarbonneauxHRS17,Radicek17,GueneauCP18,NipkowB19,MevelJP19},
to term
rewriting~\cite{AvanziniM13,AvanziniLM15,NoschinskiEG13}.
Many techniques focus on upper bounds on the worst-case bounds, but
average-case bounds~\cite{FlSaZi91,Kamin16,NgoCH17,Wang0GCQS19} and lower-bounds
have also been studied~\cite{AlbertGM12,FrohnNHBG16,NgoDFH16}.

In this paper, we extend automatic amortized resource analysis
(AARA) to cover \emph{exponential} worst-case bounds.
AARA is an effective type-based technique for deriving concrete
(non-asymptotic) worst-case bounds, in particular for functional
languages. It has been introduced by Hofmann and Jost~\cite{Jost03} to
derive \emph{linear bounds} on the heap-space usage of strict
first-order functional programs with lists.
Subsequently, AARA has been extended to programs with recursive types
and general resource metrics~\cite{Jost09}, higher-order
functions~\cite{Jost10}, lazy evaluation~\cite{VasconcelosJFH15},
parallel evaluation~\cite{HoffmannS15}, univariate polynomial
bounds~\cite{HoffmannH10}, multivariate polynomial
bounds~\cite{HoffmannAH10,HoffmannW15}, session-typed
concurrency~\cite{DasHP17}, and side
effects~\cite{LichtmanH17,NiuH18}.
However, none of the aforementioned works explores exponential bounds. 

The idea of AARA is to enrich types with numeric annotations that
represent coefficients in a potential function in the sense of
amortized analysis~\cite{tarjan85}. Bound inference is reduced to
Hindley-Milner type inference extended with linear constraints for the
numeric annotations. Advantages of the technique include
compositionality, efficient bound inference via off-the-shelf LP
solving, and the ability to derive bounds on the high-water mark for
non-monotone resources like memory.
A powerful innovation leveraged in polynomial AARA is the
representation of potential functions as non-negative linear
combinations of binomial coefficients. Their combinatorial identities
yield simple and local typing rules and support a natural semantic
understanding of types and bounds. Moreover, these potential functions
are more expressive than non-negative linear-combinations of the
standard polynomial basis.

However, polynomial potential is not always enough. Functional
languages make it particularly easy to use exponentially many
resources just by having two or more recursive calls. The following
function \texttt{subsetSum} $:$ \texttt{int list} $\to$ \texttt{int} $\to$
\texttt{bool} exemplifies this by naively solving the well-known
NP-complete problem subset sum. In the worst case, it performs
$3*2^{|nums|}-2$ Boolean and arithmetic operations (where $|x|$ gives the length of the list $x$).
\begin{lstlisting}
let subsetSum nums target =  
    match nums with
    | [] -> target = 0
    | hd::tl -> subsetSum tl (target-hd) || subsetSum tl target
\end{lstlisting}
Such a function could appear in a program with
polynomial resource usage if applied to arguments of
logarithmic size. In this case, polynomial AARA would not be able to
derive a bound. Section \ref{sec:exppoly} contains a relevant example.

To handle such functions, we introduce an extension to AARA that
allows working with potential functions of the form $f(n) = b^n$. This
extension exploits the combinatorial properties of \textit{Stirling
  numbers of the second kind} \cite{stirling} in much the same way that AARA currently
exploits those of binomial coefficients. Moreover, we allow both
multiplicative and additive mixtures of exponential and polynomial
potential functions. The techniques used in this process could easily
be applied to other potential functions in the future.

The paper first details a generalized AARA type system fit for reuse
between polynomial, exponential, and other potential functions. We
then instantiate this system with Stirling numbers of the second kind,
yielding the first AARA that can infer exponential resource
bounds. Finally, we pick out the characteristics that allow for mixing
different families of potential functions and maximizing the space
they express, and we instantiate the general system with products of
exponential and polynomial potential functions.
To focus on the main contribution, we develop the system for a simple
first-order language with lists in which resource usage is defined
with explicit \emph{tick} expressions. However, we are confident that
the results smoothly generalize to more general resource metrics,
recursive types, and higher-order functions. As in previous work, we
prove the soundness of the analysis with respect to a big-step cost
semantics that models the high-water mark of the resource usage.

\vspace{-2ex}
\section{Language and Cost Semantics}
\label{sec: Lang}
\vspace{-1.5ex}

\paragraph{Abstract Syntax}

To begin, we define an abstract binding tree (ABT, see~\cite{PFPL16}) underlying a simple strict first-order functional language. Expressions are in let-normal form to simplify the AARA typing rules. For code examples, however, we overlay the ABT with corresponding ML-based syntax. For example, $1 \!\! :: \!\! []$, $[1]$, and $\textit{cons}(1,\textit{nil})$ all represent the same list.

A program $\mathit{prog}$ is a collection of functions as defined in
the following grammar. The symbols $\mathit{lit}$, $\mathit{binop}$, and
$\mathit{unop}$ refer to standard literal values, binary operations,
and unary operations respectively, of $\mathit basic$ types
($\mathit{int}$, $\mathit{bool}$, etc.). The symbols $f$, $x$, and $r$
refer to function identifiers, variables, and rational numbers, respectively.
\vspace{-1ex}
\small
\begin{align*}
    \mathit{prog} ::= & \; \mathit{func}\{f\}(x.e) \; \mathit{prog} \mid \epsilon
    \\
    e ::=& \; \mathit{lit} \mid x \mid \mathit{binop}(x_1 ; x_2) \mid \mathit{unop}( x) \mid \mathit{app}\{f\}(x) \mid  \mathit{let}(e_1; x.e_2)
    \\
    &  \mid \mathit{share}(x_1; x_2,x_3.e) \mid \mathit{tick}\{r\} \mid \mathit{pair}(x_1; x_2) \mid \mathit{nil} \mid \mathit{cons}(x_1 ; x_2) 
    \\
    & \mid \mathit{cond}(x;e_1;e_2) \mid \mathit{pairMatch}(x_1; x_2,x_3.e) \mid \mathit{listMatch}(x_1; e_1; x_2,x_3.e_2)
\end{align*}
\normalsize
Expressions
include function applications, conditionals, and the usual introduction and elimination
forms for pairs and lists. They also include two special expressions: $\mathit{tick}\{r\}$ and $\mathit{share}$. The former, $\mathit{tick}\{r\}$, is used to specify constant resource cost $r$. We allow $r$ to be negative in the case of resources becoming available instead of being consumed. The latter, $\mathit{share}(x_1; x_2,x_3.e)$, provides two copies of its argument $x_1$ for use in $e$. This is useful because the affine features of the AARA type system do not allow naive variable reuse. In practice, $\mathit{share}$ can be left implicit by automatically preceding every variable usage by $\mathit{share}$. 

To focus on the technical novelties, we keep function identifiers and variables disjoint, that is, the types of variables do not contain arrow types and functions are first-order. Higher-order functions can be handled as in previous AARA literature \cite{HoffmannW15}. As a further simplification, we only let functions take one argument; multiple arguments can be simulated with nested pairs. Finally, the language here only supports the inductive types of lists; future work could extend this to more general types as in other AARA literature \cite{LichtmanH17,HoffmannW15,HoffmannS15JFP,HoffmannS13}.

\vspace{-2ex}
\paragraph{Operational Cost Semantics}

To define resource usage, AARA literature uses the operational big-step judgment $V \vdash e \Downarrow v \mid (q,q')$ (see e.g. \cite{Hoffmann11}) defined in Figure \ref{fig:cost}. This judgment means that, under the environment $V$, the expression $e$ evaluates to the value $v$ under some resource constraints given by the pair $q,q'$. The environment $V$ maps variables to values. The resource constraints are that $q$ is the high-water mark of resource usage, and $q-q'$ is the net amount of resources consumed during evaluation. In other words, if one started with exactly as many resources needed to evaluate $e$, that amount would be $q$, and the amount of leftover resources after evaluation would be $q'$. It is essential to track both of these values to model resources that might be returned after use, like space. Space usage usually has a positive high-water mark but no net resource consumption, as space could be reused.

The above big-step judgment only formalizes terminating evaluations. To deal with divergence, the additional judgment $V \vdash e \Downarrow \circ  \mid q$ has been introduced \cite{HoffmannH102}. This merely drops the parts of the previous judgment relevant to post-termination, focusing on partial evaluation. It means that some partial evaluation of $e$ uses a high-water mark of $q$ resources. Should it exist, the largest $q$ such that $V \vdash e \Downarrow \circ \mid q$ holds would be the high-water mark of resource usage across any partial evaluation of $e$. \ifshort{For a formal definition, see \cite{HoffmannH102}.}{The formal definition can be found in Figure \ref{fig:ntcost}.}
\begin{figure}[ht]
    \caption{Terminating operational cost semantics rules.}
    \label{fig:cost}
  \def \MathparLineskip {\lineskip=0.33cm}    
\begin{mathpar}
\fontsize{8}{12}
      	\infer[\mathit{Tick}]{
		V \vdash \mathit{tick}\{r\} \Downarrow () \mid (q,q') 
	}{
		q = max(r,0)
		&
		q' = max(-r,0)
	}

    \infer[\mathit{Binop}]{
        V \vdash \mathit{binop}(x_1, x_2) \Downarrow v \mid (0,0)
    }{
	\mathit{binop}(V(x_1),V(x_2)) \mapsto v
    }

    \infer[\!\! \mathit{Lit}]
    {
     	V \vdash \mathit{lit} \Downarrow \mathit{lit} \mid (0,0)
    }{
    }
    
    \infer[\!\!\mathit{Var}]{
        V \vdash x \Downarrow v \mid (0,0)
    }{
    	V(x) =v
    }

     \infer[\!\!\mathit{Pair}]{
        V \vdash \mathit{pair}(x_1, x_2) \Downarrow (v_1, v_2) \mid (0,0)
    }{
    	V(x_1) = v_1
	& 
	V(x_2) = v_2
    }

    \infer[\!\!\mathit{Unop}]{
        V \vdash \mathit{unop}(x) \Downarrow v \mid (0,0)
    }{
	\mathit{unop}(V(x)) \mapsto v 
    }

    \infer[\!\!\mathit{PMat}]{
    	V \vdash \mathit{pairMatch}(x_p; x_1,x_2.e) \Downarrow v \mid (q,q')
    }{
        V(x_p) = (v_1,v_2)
	& 
	V[x_1 {\mapsto} v_1, x_2 {\mapsto} v_2] \vdash e \Downarrow v \mid (q,q')
    }

    \infer[\mathit{Let}]{
        V \vdash \mathit{let}(e_1; x.e_2) \Downarrow v_2 \mid (q+max(p-q',0),p'+max(q'-p,0))
    }{
        V \vdash e_1 \Downarrow v_1 \mid (q,q')
        & 
        V[x \mapsto v_1] \vdash e_2 \Downarrow v_2 \mid (p,p')
    }
    
        \infer[\mathit{CondT}]{
            V \vdash \mathit{cond}(x_b; e_{t}; e_{f}) \Downarrow v \mid (q,q')
        }{
           	V(x_b) = \mathit{true}
		& 
		V \vdash e_t \Downarrow v \mid (q,q')
        }

        \infer[\mathit{CondF}]{
            V \vdash \mathit{cond}(x_b; e_{t}; e_{f}) \Downarrow v \mid (q,q')
        }{
           	V(x_b) = \mathit{false}
		& 
		V \vdash e_f \Downarrow v \mid (q,q')
	}

        \infer[\mathit{App}]{
            V \vdash \mathit{app}\{f\}(x) \Downarrow v  \mid (q,q')
        }{
            \mathit{func}\{f\}(x'.e) \in \mathit{prog} 
            &
            V(x) = v_x
	    &
            V[x' \mapsto v_x] \vdash e \Downarrow v \mid (q,q')
        }

        \infer[\!\!\mathit{LMat0}]{
           V \vdash \mathit{listMatch}(x; e_1; x_h,x_t.e_2)\Downarrow v \mid  (q,q')     
        }{
            V(x) = \mathit{nil}
            &
            V \vdash e_1 \Downarrow v \mid (q,q')
        }

        \infer[\!\!\mathit{Cons}]{
            V \vdash \mathit{cons}(x_h; x_t) \Downarrow v_h::v_t \mid (0,0)    
            }{
            	V(x_h) = v_h
		& 
		V(x_t) = v_t
        }

        \infer[\mathit{LMat1}]{
           V \vdash \mathit{listMatch}(x; e_1; x_h,x_t.e_2)\Downarrow v \mid  (q,q')     
        }{
            V(x) = v_h::v_t 
            &
            V[x_h \mapsto v_h, x_t \mapsto v_t] \vdash e_2 \Downarrow v \mid (q,q')
        }

        \infer[\mathit{Nil}]{
            V \vdash \mathit{nil} \Downarrow \mathit{nil} \mid (0,0)
        }{
        }

        \infer[\mathit{Share}]{
        	V \vdash \mathit{share}(x_1;x_2,x_3.e) \Downarrow v \mid (q,q')
        }{
        	V[x_2 \mapsto V(x_1), x_3 \mapsto V(x_1)] \vdash e \Downarrow v \mid (q,q')
        }
	\end{mathpar}
\end{figure}
\ifshort{}{
\begin{figure}[ht]
    \caption{Partial evaluation operational cost semantics rules.}
    \label{fig:ntcost}
    \fontsize{8}{10}
    \begin{mathpar}
    	\infer[\mathit{Partial}]{
		V \vdash e \Downarrow \circ \mid 0
	}{
	}

	\infer[\mathit{Termination}]{
		V \vdash e \Downarrow \circ \mid q
	}{
		V \vdash e \Downarrow v \mid (q,q')
	}

    \infer[\mathit{Let1}]{
        V \vdash \mathit{let}(e_1; x.e_2) \Downarrow \circ \mid q
    }{
        V \vdash e_1 \Downarrow \circ \mid q
    }

    \infer[\mathit{Let2}]{
        V \vdash \mathit{let}(e_1; x.e_2) \Downarrow \circ \mid q + \textit{max}(p-q',0)
    }{
        V \vdash e_1 \Downarrow v_1 \mid (q,q')
        & 
        V[x \mapsto v_1] \vdash e_2 \Downarrow \circ \mid p
    }

    \infer[\mathit{PMat}]{
    	V \vdash \mathit{pairMatch}(x; x_1,x_2.e) \Downarrow \circ \mid q
    }{
    	V(x) = (v_1, v_2)
	&
        V[x_1\mapsto v_1, x_2\mapsto v_2] \vdash e \Downarrow \circ \mid q
    }

        \infer[\mathit{CondT}]{
            V \vdash \mathit{cond}(x; e_{t}; e_{f}) \Downarrow  \circ \mid q
        }{
           	V(x) = \mathit{true}
		& 
		V \vdash e_t \Downarrow \circ \mid q
        }

        \infer[\mathit{CondF}]{
            V \vdash \mathit{cond}(x; e_{t}; e_{f}) \Downarrow  \circ \mid q
        }{
           	V(x) = \mathit{false}
		& 
		V \vdash e_f \Downarrow \circ \mid q
        }

        \infer[\mathit{App}]{
            V \vdash \mathit{app}\{f\}(x) \Downarrow \circ \mid  q
        }{
            \mathit{func}\{f\}(x'.e) \in \mathit{prog}
            &
            V(x) = v
            &
            V[x' \mapsto v] \vdash e \Downarrow \circ \mid  q
        }

        \infer[\mathit{LMat0}]{
           V \vdash \mathit{listMatch}(x; e_1; x_h,x_t.e_2)\Downarrow \circ \mid q 
        }{
        	   V(x) = \mathit{nil}
	    &
            V \vdash e_l \Downarrow \circ \mid q
        }

        \infer[\mathit{LMat1}]{
           V \vdash \mathit{listMatch}(x; e_1; x_h,x_t.e_2)\Downarrow  \circ \mid q  
        }{
            V(x) = v_h::v_t
            &
            V[x_h \mapsto v_h, x_t \mapsto v_t] \vdash e_2 \Downarrow \circ \mid q
        }

    	\infer[\mathit{Tick}]{
		V \vdash \mathit{tick}\{r\} \Downarrow \circ \mid \mathit{max}(r,0)
	}{
	}
	\end{mathpar}
\end{figure}
}

\vspace{-2ex}
\section{Automatic Amortized Resource Analysis}
\vspace{-1.5ex}

Here we lay out a generalized version of the AARA system with the potential functions abstracted. Existing AARA literature is specialized to polynomial functions (see e.g. \cite{HoffmannH10}). This existing polynomial system may be obtained as an instantiation, as may the exponential system that we introduce in Section \ref{sec: exppot}.

AARA uses the \textit{potential} (or physicist's) method to account for resource use, as is commonly used in amortized analyses. The potential method uses the physical analogy of converting between potential and actual energy that can be used to perform work. Whereas a physicist might find potential in the chemical bonds of a fuel, however, AARA places it in the constructors of lists.

To prime intuition with an example, consider paying a resource for each $::$ operation performed in the following code. It performs $\mathit snoc$, which is like $\mathit cons$ but adds onto the back of the list rather than the front.
\begin{lstlisting}[mathescape=true]
let snoc x xs = 
    match xs with
    | [] -> tick 1; x::[]	                (* pay 1 resource here *)
    | hd::tl -> tick 1; hd::(snoc x tl)	(* pay 1 resource here *)
\end{lstlisting}

The resource consumption of $\mathit{snoc}\; x\; xs$ as defined by the $\mathit tick$ expressions is $1 + |xs|$.
Using the potential method, we can justify this bound as follows.
If 1 resource is initially available, then the base case of the empty list can be paid for. 
If there is $1$ stored per element of the list then $1$ resource is released in the cons case of the pattern match.
This suffices to pay for the additional $::$ operation.
The remaining potential on $xs$ can be assigned to $tl$ for the recursive call.
One can sum these costs to infer that the initial potential $1 + |\mathit{xs}|$ covers the cost of all the $::$ operations. The AARA type system could describe this with the typing $L^1(\mathbb Z)$ for $\mathit xs$ (describing the linear potential in the superscript) and $\mathbb Z \times L^1(\mathbb Z) \stackrel {1/0}{\rightarrow} L^0(\mathbb Z)$ for $\mathit snoc$ (describing the initial/remaining resources above the arrow).
Another valid type is $\mathbb Z \times L^2(\mathbb Z) \stackrel {1/0}{\rightarrow} L^1(\mathbb Z)$, which could be used
in a context where the result of $\mathit snoc$ must be used to pay for additional cost.

\vspace{-2ex}
\paragraph{Types}

The AARA system laid out here supports the types given below. The symbol $F$ gives the types of functions, where $q$ and $q'$ are non-negative rationals. The symbol $S$ gives the remaining non-function types, where $\mathit{basic}$ stands for the basic types like $\mathit{int}$ or $\mathit{unit}$, and the resource annotation $P$ is an indexed family of rationals representing the coefficients in a linear combination of basic potential functions. 
\vspace{-1.5ex}
\begin{align*}
    F ::= \;  S \stackrel{q/q'}{\rightarrow} S 
    & &
    S ::= \; \mathit{basic} \mid L^P(S) \mid S\times S
\end{align*}
The typing rules for these types are given in Figure \ref{fig: type rules} and explained in the following sections. The values of these types are the usual values.

\vspace{-2ex}
\paragraph{Potential}

To understand typing rules, it is necessary to define potential. The following potential constructs are generalized from  polynomial AARA work \cite{HoffmannH10}.

As mentioned, $P = (p_i)_{i \in I}$ is in $\mathbb{Q}^I$ as an indexed family of rationals. Each entry represents a coefficient in a linear combination of basic potential functions. This linearity makes it natural to overload the type of $P$ as a vector or matrix of rationals, so it is treated as such whenever the context is appropriate. Finally, let those basic potential functions be fixed as some family $(f_i)_{i\in I}$, where $f_i(0)=0$.

We define the potential represented with $P$ using the function $\phi$ where
\vspace{-.5ex} 
\begin{center}
$\phi(n, P) = \sum_{i} p_i \cdot f_i(n)\;.$  
\end{center}
\vspace{-.5ex} 
The function $\phi$ yields the total potential on a list (excluding the potential of its elements) as a function of the list's size $n$ and its potential annotation $P$.

We can then relate resource potential between different sizes of list with the shift operator $\lhd: \mathbb{Q}^I \rightarrow \mathbb{Q}^I$ and constant difference operator $\delta:\mathbb{Q}^I \rightarrow \mathbb{Q}$. These functions need only satisfy the following property equation.
\vspace{-.5ex} 
\begin{equation}
  \label{eq:shift}
\phi(n+1, P) = \delta(P) + \phi(n, \lhd P)
\end{equation}
\vspace{-.5ex} 
Though we leave open the explicit definition of these functions for generality, we only later work with instances of them that are linear operators, such that Equation~\ref{eq:shift} denotes a linear recurrence. Such a refinement leaves $\lhd P$ and $\delta(P)$ linear functions of $P$.

These functions come in handy for understanding the stored potential in a value of a given type, defined by the potential function $\Phi$ as follows. 
\vspace{-.5ex} 
\begin{align*}
        \Phi(v:\mathit{basic}) &= 0 
        \\
        \Phi((v_1,v_2):A_1 \times A_2) &= \Phi(v_1:A_1) + \Phi(v_2:A_2)
        \\
        \Phi([]:L^P(A)) &= 0
        \\
        \Phi(h::t : L^P(A)) &= \delta(P) + \Phi(h:A) + \Phi(t:L^{\lhd P}(A))
    \end{align*}

We often need to measure the potential across an entire evaluation context of typed values $V:\Gamma$ given by a typing context $\Gamma$ and variable bindings $V$. We do so by extending the definition of potential $\Phi$ as follows.
\vspace{-.5ex} 
\begin{center}
    $\Phi(\emptyset) = 0 \hspace{4em} \Phi(V : (\Gamma,v : A)) = \Phi(V:\Gamma) + \Phi(v:A)$
\end{center}

Finally, we can use these definitions to obtain a closed-form expression for the potential over an entire list (including its elements) with the following:

\begin{lemma}
    \label{thm: closed pot}
    Let $l=[a_n,...,a_1]$ be a list of $n$ values. Then $
        \Phi(l:L^P(A)) = \phi(n,P) + \sum_{i=1}^{n}\Phi(a_i:A)$
\end{lemma}

\begin{proof}
    We induct over the structure of the list $l$.
    
    For the empty list of length 0:
    \vspace{-1ex}
        $$\Phi([]:L^P(A)) = 0 ={\scriptstyle\sum}_i p_i \cdot f_i(0) = \phi(0,P) + {\scriptstyle\sum}_{i=1}^{0}\Phi(a_i:A)$$
    \vspace{-4ex}
    
    For $l = h::t$ of size $n+1$: 
    \vspace{-1ex}
    \begin{align*}
        \Phi(a_{n+1}::b : L^P(A)) &= \delta( P) + \Phi(a_{n+1}:A) + \Phi(l':L^{\lhd P}(A))
        \\
        &= \delta(P) + \Phi(a_{n+1}:A) + \phi(n,\lhd P) + {\scriptstyle\sum}_{i=1}^{n}\Phi(a_i:A) 
        \\
        &= \phi(n+1,P) +{\scriptstyle\sum}_{i=1}^{n+1}\Phi(a_i:A) 
    \end{align*}
\end{proof}

We can apply Lemma \ref{thm: closed pot} to the previously defined function $\mathit snoc$ to see the change in potential between input and output. This difference in potential should bound the resources consumed. For this case, the basic potential functions $(f_i)$ only need contain $\lambda n.n$, and we can let $\lhd (p) = p = \delta((p))$. Letting $y$ be the result of $snoc \; x \; xs$, the type $\mathbb{Z} \times L^1(\mathbb{Z}) \stackrel {1/0}{\rightarrow} L^0(\mathbb{Z})$ indicates the following bound  
$$\Phi(x:\mathbb{Z}, xs:L^1(\mathbb{Z})) + 1 - \Phi(y : L^0(\mathbb{Z})) = \phi(|xs|, 1) + 1 - \phi(|y|, 0) = |xs| + 1 $$

This is exactly the amount of resources consumed, so the bound is tight.

In this work we only consider so-called \textit{univariate} potential, wherein every term in the potential sum is dependent on the length of only one input list. However, different univariate potential summands may depend on different inputs, and thus univariate potential may still be multivariate. The term \textit{multivariate potential} refers to using more general multivariate functions for potential. There is existent work on multivariate potential using polynomial functions \cite{HoffmannAH12}. We expect that the work here extends to multivariate potential similarly. 

\begin{figure}[p]
    \centering
    \caption{AARA typing rules.}
    \label{fig: type rules}
\def \MathparLineskip {\lineskip=0.33cm}    
\begin{mathpar}
\fontsize{8}{10}
	\textbf{Basic rules:}
	\\
    \infer[\mathit{Lit}]
    {
        \Sigma; \emptyset \qprove 0 0 \mathit{lit}:\mathit{basic}
    }{
    }

    \infer[\mathit{Let}]{
        \Sigma; \Gamma_1,\Gamma_2 \qprove q {q'} \mathit{let}(e_1; x.e_2):B
    }{
        \Sigma; \Gamma_1 \qprove q {p} e_1:A 
        & 
        \Sigma; \Gamma_2, x:A \qprove p {q'} e_2:B
    }

    \infer[ \mathit{Unop}]{
        \Sigma; x:basic \qprove 0 0 \mathit{unop}(x): \mathit{basic'}
    }{
    }

     \infer[\mathit{Binop}]{
        \Sigma; x_i:\mathit{basic} \qprove 0 0 \mathit{binop}(x_1, x_2): \mathit{basic'}
    }{
    }

    \infer[\mathit{Var}]{
        \Sigma; x:A \qprove 0 0 x:A
    }{
    }

    \infer[\mathit{Pair}]{
        \Sigma; x_1:A_1, x_2:A_2 \qprove 0 0 \mathit{pair}(x_1,x_2):A_1\times A_2
    }{
    }

    \infer[\mathit{PMat}]{
        \Sigma;\Gamma, x:A_1\times A_2 \qprove q {q'} \mathit{pairMatch}(x; x_1,x_2.e):B
    }{
        \Sigma;\Gamma, x_1:A_1,x_2:A_2 \qprove q {q'} e:B
    }

        \infer[\mathit{Cond}]{
            \Sigma; \Gamma, x:\mathit{bool}\qprove q {q'} \mathit{cond}(x; e_1; e_2):A
        }{
            \Sigma; \Gamma, x:\mathit{bool} \qprove q {q'}e_1:A
            &
            \Sigma; \Gamma, x:\mathit{bool} \qprove q {q'} e_2:A
        }
        \\
        \\
	\textbf{Function \;rules:}
	\\
        \infer[\mathit{App}]{
            \Sigma; x:A \qprove q {q'} \mathit{app}\{f\}(x) : B
        }{
            A \stackrel{q/q'}{\rightarrow} B \in \Sigma(f)
        }

        \infer[\mathit{Fun}]{
            A \stackrel{q/q'}{\rightarrow} B \in \Sigma(f)
        }{
            \mathit{func}\{f\}(x.e) \in \mathit{prog} 
            &
            \Sigma; x:A \qprove q {q'} e:B
        }
	\\
	\\
	\textbf{Potential-focused\;rules:}
	\\
        \infer[\mathit{Tick}]{
            \Sigma; \Gamma \qprove{max(r,0)} {max(-r,0)} \mathit{tick}\{r\}:unit
        }{
        }
     
        \infer[\mathit{Relax}]{
            \Sigma; \Gamma \qprove q {q'} e:A
        }{
            \Sigma; \Gamma \qprove p {p'} e:A
            &
            q \geq p
            &
            q-p \geq q'-p'
        }

        \infer[\mathit{SubWeakL}]{
            \Sigma;\Gamma, x:A' \qprove q {q'} e:B
        }{
            \Sigma;\Gamma, x:A \qprove q {q'} e:B
            &
            A' <: A
        }

        \infer[\mathit{SubWeakR}]{
            \Sigma;\Gamma \qprove q {q'} e:A
        }{
            \Sigma;\Gamma \qprove q {q'} e:A'
            &
            A' <: A
        }

        \infer[\mathit{Sharing}]{
            \Sigma;\Gamma, x_1:A_1 \qprove q {q'} \mathit{share}(x_1; x_2,x_3.e):B
        }{
            \Sigma; \Gamma, x_2:A_2, x_3:A_3 \qprove q {q'} e:B
            &
            A_1 \curlyvee (A_2, A_3)
        }
        \\
        \\
        \textbf{List\;rules:}
        \\
        \infer[\mathit{Nil}]{
            \Sigma; \emptyset \qprove 0 0 \mathit{nil} : L^P(A)
        }{
        }

        \infer[\mathit{Cons}]{
            \Sigma; x_h:A, x_t:L^{\lhd P}(A) \qprove {\delta(P) } 0 \mathit{cons}(x_h; x_t): L^P(A)
        }{
        }

        \infer[\mathit{ListMatch}]{
            \Sigma;\Gamma, x:L^P(A) \qprove q {q'} \mathit{listMatch}(x; e_1; x_h,x_t.e_2):B
        }{
            \Sigma;\Gamma \qprove q {q'} e_1:B 
            &
            \Sigma; \Gamma, x_h:A, x_t: L^{\lhd P}(A) \qprove {q+\delta(P)} {q'} e_2:B
        }
     \end{mathpar}
 \end{figure}

 \begin{figure}[ht]
    \caption{AARA subtyping and sharing judgments.}
    \label{fig:judgments}
\def \MathparLineskip {\lineskip=0.2cm}    
\begin{mathpar}
\fontsize{8}{10}
        \infer[\mathit{Subtype}]{
            L^P(A) <: L^Q(A)
        }{
            \forall i. p_i \geq q_i
        }

        \infer[\mathit{ShareBasic}]{
            \mathit{basic} \curlyvee(\mathit{basic}, \mathit{basic})
        }{
        }

        \infer[\mathit{SharePair}]{
            A_1 \times B_1 \curlyvee (A_2\times B_2,A_3\times B_3)
        }{
            A_1 \curlyvee (A_2, A_3)
            & 
            B_1 \curlyvee (B_2, B_3)
        }
      
         \infer[\mathit{ShareList}]{
            L^P(A_1) \curlyvee (L^Q(A_2), L^R(A_3))
        }{
            A_1\curlyvee (A_2, A_3)
            &
            P=Q+R
        }
\vspace{-5ex}
    \end{mathpar}
\end{figure}

\vspace{-2ex}
\paragraph{Typing Rules}

The typing rules in Figure \ref{fig: type rules} use the judgment $\Sigma; \Gamma \qprove q {q'} e:A$. In this typing judgment, $\Gamma$ maps variables to types, while $\Sigma$ maps function labels to sets of types. This judgment holds when, in the typing environment given by $\Sigma$ and $\Gamma$, the expression $e$ is of type $A$, subject to the constraints that $q$ and $q'$ are the amount of available resources before and after some evaluation of $e$. Unlike the judgment $V \vdash e \Downarrow v \mid (q,q')$, these values need not be tight. 

By expressing available resources on the turnstile, and potential resources in the types given by $\Sigma,\Gamma$, and $A$, the type system is set up to formalize the reasoning of the potential method. Theorem \ref{thm: sound} shows that it is sound with respect to the operational semantics of Section \ref{sec: Lang}. 

Many typing rules preserve the total resource potential they are given, consuming none of it themselves. They therefore usually either have no explicit interaction with potential (e.g. $\mathit{Lit}$) or pass around exactly what they are given (e.g. $\mathit{Let}$). All basic rules in the first block of Figure \ref{fig: type rules} fit this characterization.

The typing rules concerning functions in second block of Figure \ref{fig: type rules} are the only to make use of $\Sigma$. For each function $f$ defined in $\mathit prog$ via $\mathit func\{f\}(x.e)$, $\Sigma(f)$ refers to the set of types that its body $e$ could be given. That we allow for sets of types is important because recursive calls to a function may not always make use of a type with the same resource annotations; this is called \textit{resource-polymorphic recursion}. Despite these rules capturing the intuition behind typing resource-polymorphic recursion, they are not used in existing implementation, as they lead to infinite type derivations. Nonetheless there exists an effective way to type resource-polymorphic recursion with a finite derivation; see \cite{HoffmannH102}. In the examples provided in this article, it usually suffices to consider only \textit{resource-monomorphic recursion}, wherein inner and outer calls use the same annotation.

All of the rules discussed so far are simply those of existing AARA literature with their parameter for operation cost set to 0 (see e.g. \cite{HoffmannH10}). This does not change their generality, as such constant cost can (and could already in prior work) be simulated using $\mathit{tick}$. Similarly, non-constant costs could be simulated by running helper functions using $\mathit{tick}$ the appropriate number of times. 

The remaining rules cover sharing, subtype-weakening, and the rules concerning lists. Weakening, though not listed, is also allowed.

Sharing is a form of contraction. By sharing, the rest of the typing rules can become affine, allowing only single usages of a given variable. Intuitively, sharing is meant to prevent duplicating potential across multiple usages of a variable, and instead split the potential across them. The rules for the sharing judgment, indicating how to split potential, can be found in Figure \ref{fig:judgments}. Note that the rule $\textit{ShareList}$ adds indexed collections of rationals; this should be interpreted pointwise, as if the addends were vectors or matrices.

Subtype-weakening is a form of subtyping based on potential. It discards potential on a list, weakening the upper bound on resources it represents. This rule follows all usual subtyping rules, as well as $\mathit {Subtype}$ from Figure \ref{fig:judgments}. Relaxing behaves similarly, but loosens the bounds on the available resources instead.

The intuition for the rules concerning lists in the last block of Figure \ref{fig: type rules} is that total resources should be conserved between constructions and destructions. Because $\delta(P)$ expresses the difference in potential, it is exactly how many resource units are released after a pattern match on a list of type $L^P(A)$. For the same reason, it is also how many need to be stored when reversing the process and putting an element on a list of type $L^{\lhd P}(A)$. Finally, when a list is empty, it has no room to store potential. Every potential function $f_i$ maps 0 to 0, so an empty list can safely be assigned any scalar of zero potential.

\vspace{-2ex}
\paragraph{Soundness}

The soundness of the type system is expressed with the following theorem. It states that the evaluation of an expression $e$ does not require more resources than initially present, and (should evaluation terminate) it leaves at least as many resource as dictated. The proof is a straightforward generalization of the version from \cite{HoffmannH10}, but we nonetheless reproduce the proof below. 

\vspace{-1.5ex} 
\begin{theorem}
	\label{thm: sound}
    Let $\Sigma; \Gamma \qprove q {q'} e:B$ and $V$ provide the variable bindings for $\Gamma$    \begin{enumerate}
        \item If $V \vdash e \Downarrow v \mid (p,p')$ then $p \leq \Phi(V:\Gamma) +q$ and $p-p' \leq \Phi(V:\Gamma) + q - \Phi(v:B) - q'$
        
        \item If $V \vdash e \Downarrow \circ \mid p$ then $p \leq \Phi(V:\Gamma) +q$ 
    \end{enumerate}
\end{theorem}
\vspace{-2.5ex} 

\begin{proof}
Assume $V$ binds $\Gamma$'s variables and perform nested induction on the type derivation and operational judgment for an expression in let-normal form. We show the induction below only for the terminating operational judgment cases, but the partial-evaluation cases are nearly identical.

(\textbf{Base Non-Cons}) Suppose the last rule applied in the typing derivation is any non-$\mathit{Cons}$ base case, i.e., $\mathit{Lit}$, $\mathit{Var}$, $\mathit{Unop}$, $\mathit{Binop}$, $\mathit{Pair}$, $\mathit{Nil}$, or $\mathit{Tick}$. Then assume the appropriate terminating operational judgment rule applies. In such a case, one finds $p \leq q$, $p' \geq q'$, and $\Phi(v:B)=\Phi(V:\Gamma)$. This and the non-negativity of potential are sufficient to satisfy the desired inequalities.

(\textbf{Base Cons}) Suppose the last rule is $\mathit{Cons}$, so $q=\delta(P)$ and $q'=0$. Assume the $\mathit{Cons}$ operational judgment applies, so that $p = p' = 0$. Note $\Phi(v_h::v_t : L^P)$ is equal to $\delta(P) + \Phi(v_h:A) + \Phi(v_t: L^P(A))$ by definition. This identity and the non-negativity of potential satisfy the desired inequalities.

(\textbf{Step Implicit Inequalities}) Suppose the last rule is one of $\mathit{SubWeakL}$, $\mathit{SubWeakR}$, $\mathit{Relax}$, or substructural weakening, and assume some operational judgment applies. Each typing requires a similar typing judgment as a premiss. Further, none changes any values, so the same operational judgment still applies. Thus, the inductive hypothesis applies, and gives almost the inequalities we need. Each case provides the inequalities needed to finish. For subtype-weakening, it is sufficient note that $C <: D$ entails $\Phi(v:C) \geq \Phi(v:D)$, since $C$ is pointwise greater-then-or-equal to $D$. For $\mathit{relax}$, the premisses of the $\mathit{relax}$ rule directly include the inequalities needed to complete the case. And we can complete the substructural weakening case by noting that the non-negativity of potential entails $\Phi(V:\Gamma,v:A) \geq \Phi(V:\Gamma)$. 

(\textbf{Step Let}) Suppose the last rule is $\mathit{Let}$, and suppose its operational judgment applies. The premisses of the typing rule require that $\Sigma; \Gamma_1 \qprove q {r} e_1:A$ and $\Sigma; \Gamma_2, x:A \qprove r {q'} e_2:B$. The premisses of the operational judgment require that $V \vdash e_1 \Downarrow v_1 \mid (s,s')$ and $V[x \mapsto v_1] \vdash e_2 \Downarrow v_2 \mid (t,t')$, where $p = s+max(t-s',0)$ and $p'=t'+max(s'-t,0)$. Applying the inductive hypothesis to these premiss pairs and adding the resulting inequalities cancels terms to complete the case.

(\textbf{Step Sharing}) Suppose the last is $\mathit{Sharing}$, so that $\Gamma = \Gamma',x_1:A_1$. It requires as a premiss that $\Sigma; \Gamma', x_2:A_2, x_3:A_3 \qprove q {q'} e:B$, where $A_1 \curlyvee (A_2, A_3)$. Assuming the operational judgment $\mathit{Share}$ applies, $V[x_2 \mapsto V(x_1), x_3 \mapsto V(x_1)] \vdash e \Downarrow v \mid (p,p')$ also holds. The inductive hypothesis applies, yielding the needed inequalities, but for $x_2,x_3$ instead of $x_1$. However, the sharing relation ensures that $\Phi(v_1:A_1) = \Phi(v_2:A_2,v_3:A_3)$, and this identity finishes the case.

(\textbf{Step ListMatch}) Suppose the last is $\mathit{ListMatch}$, so $\Gamma = \Gamma',x:L^P(A)$. There are two operational judgments which could apply: $\mathit{LMat0}$ and $\mathit{LMat1}$.

Suppose the former judgment applies. It requires that $V \vdash e_1 \Downarrow v \mid (p,p')$. At the same time, the $\mathit{ListMatch}$ rule requires as a premiss that $\Sigma; \Gamma' \qprove q {q'} e_1:B$. The inductive hypothesis applies, yielding the needed inequalities, but for $\Gamma'$ instead of $\Gamma$. However, because $\Phi(nil:L^P(A))=0$, we see $\Phi(V:\Gamma') = \Phi(V:\Gamma)$, and the desired inequalities result.

Suppose instead the latter judgment applies. This judgment requires as a premiss that $V[x_h \mapsto v_h, x_t \mapsto v_t] \vdash e_2 \Downarrow v \mid (p,p')$. At the same time, the $\mathit{ListMatch}$ rule requires that $\Sigma; \Gamma', x_h :A, x_t :L^{\lhd P}(A) \qprove {q+\delta(P)} {q'} e_2:B$. The inductive hypothesis applies, telling us that  $p-p' \leq \Phi(V:\Gamma', v_h:A, v_t:L^{\lhd P}(A)) + q+\delta(P) - \Phi(v:B) - q'$ and $p \leq \Phi(V:\Gamma', v_h:A, v_t:L^{\lhd P}(A)) +q+\delta(P)$. By definition, $\Phi(v_h::v_t : L^P) = \delta(P) + \Phi(v_h:A) + \Phi(v_t: L^P(A))$, and applying this identity to the inequalities yields the inequalities needed.

(\textbf{Step Cond}) Suppose the last rule is $\mathit{Cond}$, and that either of the $\mathit{CondT}$ or $\mathit{CondF}$ operational judgments apply. In either case, applying the inductive hypothesis to its premiss and the premiss of $\mathit{Cond}$ gives the needed inequalities.

(\textbf{Step PMat}) Suppose that the last rule applied is $\mathit{PMat}$, so that $\Gamma = \Gamma',x:A_1 \times A_2$. This rule would require as a premiss that $\Sigma; \Gamma', x_1:A_1, x_2:A_2 \qprove q {q'} e' :B$, for $e'$ the body of the match statement $e$. Suppose the $\mathit{PMat}$ operational judgment applies. This judgment requires as a premiss that $V[x_1 \mapsto v_1, x_2 \mapsto v_2] \vdash e' \Downarrow v \mid (p,p')$, where the value of $x$ is $(v_1,v_2)$. Applying the inductive hypothesis to these premisses followed by the definitional identity $\Phi((v_1,v_2):A_1 \times A_2) = \Phi(v_1:A_1) + \Phi(v_2:A_2)$ completes the case.

(\textbf{Step App}) Suppose the last rule is $\mathit{App}$. Note that this rule requires $\mathit{Fun}$ as a premiss, which in turn requires $\Sigma; x:A \qprove q {q'} e' :B$ where $e'$ is the body of the function being applied. If the $\mathit{App}$ operational judgment applies, its premiss would require $V[x' \mapsto V(x)] \vdash e \Downarrow v \mid (p,p')$. Although $e'$ might not be a smaller expression than $e$, the operational judgment derivation still shrinks.  This means the inductive hypothesis applies, and it gives the exact inequalities needed.

\end{proof}

\vspace{-2ex}
\paragraph{Type Inference}

Type inference for the Hindley-Milner part of the type system is decidable \cite{Hindley,Milner78atheory}. The only new barrier for automating inference in AARA is obtaining witnesses for all the coefficients in each annotation $P$ in a derivation. 

Each typing rule naturally gives a set of linear constraints on the entries of $P$. If the relation given by $\lhd$ and $\delta$ can likewise be expressed with linear constraints, then all such constraints are linear. So long as $|P|$ is finite, this forms a linear program. A linear program solver can then find minimal witnesses efficiently.

Existing AARA literature (see e.g. \cite{HoffmannH10}), however, uses binomial coefficients as the basis functions for $P$, of which there are infinitely many. This nonetheless works because only a particular finite prefix of their set, $\binom - 1, \dots , \binom - k$, are used as a basis in a given analysis. Each such prefix basis also yields the same locally-definable shift operation: the linear equality $\lhd p_i = p_i + p_{i+1}$, where $p_k$ is the coefficient of $\binom - k$ and is 0 if the function is outside the prefix. As this is a linear relation, and each prefix is finite, inference can be performed via linear program. The prefix bases of binomial coefficients thereby form an infinite family of finite bases, each of which allows automated inference of resource polynomials up to a fixed degree in the AARA system.

As a caveat, not all programs use resources in a manner compatible with the AARA system. Indeed, it is undecidable whether or not a program uses e.g. polynomial amounts of resources, as this could solve the halting problem.

\vspace{-2ex}
\section{Exponential Potential}
\label{sec: exppot}
\vspace{-1.5ex}

Stirling numbers of the second kind $\stirling n k = \frac 1 {k!} \sum_{i=0}^k (-1)^i \binom k i (k-i)^n$ count the number of ways to form a $k$-partition of a set of $n$ elements. These can be used to express exponential potential functions similarly to how binomial coefficients can express polynomial ones. In particular, we make use of Stirling numbers with arguments $n,k$ offset by 1, $\stirling {n+1} {k+1}$, so that $\phi(n,P) = \sum_i p_i \cdot \stirling {n+1}{i+1}$. While other bases could also express exponential potential, these offset Stirling numbers have a few particularly desirable properties, which are described in this section.

\vspace{-2ex}
\paragraph{Simple Shift Operation}

Like binomial coefficients, the prefixes of the basis of the offset Stirling numbers of the second kind form an infinite family of finite bases, each of which allows automated inference in the AARA system. However, these potential functions are exponential rather than polynomial.

Stirling numbers of the second kind satisfy the recurrence $\stirling {n+1} {k+1} = (k+1)\stirling n {k+1} + \stirling n k$. This recurrence allows the $\lhd$ operation to have the same local definition for every annotation entry in every prefix basis: $\lhd p_i = (i+1)p_i + p_{i+1}$, where $p_k$ is the coefficient of $\stirling {n+1}{k+1}$, and is 0 if the function index is outside the chosen prefix. Given this definition for $\lhd$ and letting $\delta(P) = p_0$, we find $p_0 + \sum_i \lhd p_i \stirling {n+1}{i+1} = \sum_i p_i \stirling{n+2}{i+1}$, satisfying Equation~\ref{eq:shift}.

This shift operation yields a linear relation, as the coefficient of a given $p_i$ is a constant scalar. Thus, exactly like when using binomial coefficients, inference is automatable via linear programming. Certain other exponential bases, like Gaussian binomial coefficients, could be similarly automated.

\vspace{-2ex}
\paragraph{Expressivity}

Because $\stirling {n+1}{k+1} = \frac 1 {k!} \sum_{i=0}^k (-1)^{k-i}\binom k i (i+1)^n \in \Theta((k+1)^n)$, the offset Stirling numbers of the second kind can form a linear basis for the space of sums of exponential functions. Each function $\lambda n.b^n$ with $b\geq 1$ can be expressed as a linear combination of the functions $\lambda n.\stirling {n+1}{k+1}$.

The function $\lambda n.\stirling {n+1}{k+1}$ is also non-negative for natural $n$, and non-decreasing with respect to $n$. These are two natural properties to require of basic potential functions, since amortized analysis requires non-negative resources, and larger inputs should not usually become cheaper to process. Further, the properties are preserved by non-negative linear (i.e. conical) combination, and by $\lhd$ when defined with a non-negative linear recurrence - the combinations given by $P$ and $\lhd P$ always satisfy the two potential function properties.

Ensuring these properties for more general potential functions requires determining if such a function on a natural domain is always non-negative. This is non-trivial. In the existing literature on multivariate polynomials, we find this is \textit{undecidable} in the worst case \cite{matiyasevich}. However, restricting to non-negative linear (that is, \textit{conical}) combinations of non-negative, non-decreasing functions - as we have done here - gives simple linear constraints that ensure both desired properties. For finite bases, this is easily handled via linear programming. 

When considering expressivity in this conical combination model of potential functions, one finds some otherwise-valid potential functions are not be expressible in the conical space given by the offset Stirling number functions. Nonetheless, Stirling number functions are a \textit{maximally expressive} basis; it is not possible to express additional potential functions using a different basis without losing expressibility elsewhere. Notably, the standard exponential basis is \textit{not} maximal in this sense. The formal statement of such maximal expressivity is generalized in the theorem below. Any finite, sequential subset of the offset Stirling number functions satisfy the prerequisites of this theorem, as do the binomial coefficient functions and other well-known functions like the Gaussian polynomials.

\begin{theorem}
    \label{thm:max}
    Let $\{f_i\}$ be a finite set of linearly independent functions on the naturals that are non-negative and non-decreasing. Let $f_i(n)$ be 0 until $n \geq i$, and let $i \leq j$ imply that $O(f_i )\subseteq O(f_j)$, with asymptotic equality only when $i=j$. Let $L$ be the linear span (collection of linear combinations) of $\{f_i\}$, and let $C$ be its conical span (collection of conical combinations). 
    
    There does not exist another linearly independent basis $\{g_i\}$ with linear span $L$ and conical span $D \supsetneq C$ such that each function in $\{g_i\}$ is non-negative and non-decreasing. That is, $\{f_i\}$ has a maximally expressive conical span.
\end{theorem}
 \vspace{-2ex}
\begin{proof}
    Suppose there is such a basis $\{g_i\}$. We express each basis $\{f_i\}$ and $\{g_i\}$  with linear combinations of the other, and derive a contradiction.
    
    If there is any function in the conical span $D$ of $\{g_i\}$ that is not in $C$, then this is the case for some basis function $g_k$. Because $g_k \in L$, it can be written as a linear combination of $\{f_i\}$; let $\sum_i \alpha_i f_i = g_k$. Because $g_k \not \in C$, there is at least one coefficient $\alpha_i < 0$; let it be $\alpha_m$. In case there are multiple candidate elements $g_k$, pick $g_k$ to be the basis function such that this index $m$ is minimized.
    
    We then see that $
        g_k(m) = \sum_i \alpha_i f_i(m) = (\sum_{i < m}\alpha_i f_i(m)) + \alpha_m f_m(m)$ because $f_i(m)$ for $i>m$ is 0. This yields two observations: First, $m<k$, as otherwise the fastest-growing term of $g_k$ would be negative, but $g_k$ is never negative. Second, the term $\alpha_m f_m(m)$ is negative, yet $g_k \geq 0$, so it must be that $\sum_{i < m}\alpha_i f_i(m) > 0$. Thus there exists a coefficient $\alpha_p > 0$ where $p < m$. 
    
    Now we look at representing $\{f_i\}$ with $\{g_i\}$. Because the conical span $D$ contains $C$, it can represent each $f_i$ as a conical combination. Notably, a given $f_i$ cannot be represented only with functions outside of $\Omega (f_i)$, nor any function outside of $O(f_i)$, due to growth rates. There is therefore at least one function in $\{g_i\}$ that is $\Theta(f_i)$, for each $i$. Since the linear span of these corresponding $g_i$ already has the same (finite) dimension as $L$, any additional functions would not be linearly independent. Due to this, we can say $g_i \in \Theta(f_i)$ uniquely for each $i$. 
        
    Take $f_k$ in particular as a conical combination of $\{g_i\}$. We now consider replacing each element of $\{g_i\}$ in that conical combination with its equivalent linear combination of elements of $\{f_i\}$. Because of the above correspondence of growth rates, there must be a positive coefficient for $g_k$. Because $g_k$ has positive weight $\alpha_p$ on $f_p$ where $p<m<k$, another basis function $g_i$ in the conical combination must have negative weight on $f_p$ to cancel it out in their linear combination. However, $g_k$ was picked such that it had the lowest index $m$ with negative weight across all $\{g_i\}$; it is contradictory for there to be such a $p < m$.
\end{proof}
 
 \vspace{-2ex}
\paragraph{Natural Semantics}

The values of $\stirling {n+1} {k+1}$ count the number of ways to pick $k$ non-empty disjoint subsets of $n$ elements. Many programs with exponential resource use iterate over collections of subsets, so these numbers naturally arise. 

Recall the naive solution to subset sum from the introduction. The algorithm iterates through all the subsets of numbers in the input list. When considering Fagin's descriptive complexity result that NP problems are precisely those expressible in existential second order logic \cite{Fagin}, it becomes clear that naive solutions to any NP-complete problem fit this characterization: naively brute-forcing through second order terms to find an existential witness is just iterating through tuples of subsets.

\vspace{-2ex}
\paragraph{Example}

Consider the naive solution to subset sum from the introduction. One can verify that the number of Boolean and arithmetic operations used on an input of size $n$ is $3*2^n-2$ by induction. We find the same bound here by preceding each such operation with an explicit $\mathit{tick}\{1\}$ operation. Thee AARA type system then verifies that the type of $\mathit{subsetSum}$ is $L^3(\mathbb{Z})\times \mathbb{Z} \stackrel{1/0}{\rightarrow} \mathit{bool}$.

Here is the code again, with type annotations on each line tracking the amount of $\stirling {n+1} 2$ potential on lists, and comments tracking available constant potential. For clarity, the code is re-written in a let-normal form, and sharing locations are marked. 
\vspace{-1.5ex}
\begin{lstlisting}[mathescape=true]
let subsetSum nums:$L^3(\mathbb{Z})$ target =                         	   (* 1 *)
    match nums:$L^3(\mathbb{Z})$ with
    | [] ->                                          	   (* 1 *)
    	tick 1; target = 0 		     		   (* 0 *)
    | hd::(tl:$L^6(\mathbb{Z})$) ->                                        (* 4 *)
        tick 1; let newTarget = target - hd in             (* 3 *)
        (* share tl:$L^6(\mathbb{Z})$ as $L^3(\mathbb{Z})$, $L^3(\mathbb{Z})$ *)
        let withNum = subsetSum tl:$L^3(\mathbb{Z})$ newTarget in  (* 2 *)
        let without = subsetSum tl:$L^3(\mathbb{Z})$ target in              (* 1 *)
        tick 1; withNum || without                   	   (* 0 *)
\end{lstlisting}
\vspace{-1.5ex} 

The indicated values yield witnesses for the AARA typing rules, so we know via soundness that the difference between initial and ending potential gives an upper bound on how many operations were used. That difference is $1+3*\stirling {n+1} 2 = 3*2^n-2$, where $n$ is the size of $\mathit{nums}$, exactly the amount used.

Exponential terms with higher bases than 2 can come into play with more recursive calls, like in the code below enumerating the $3^n$ ways to put $n$ labelled balls into 3 labelled bins.
\vspace{-1.5ex} 
\begin{lstlisting}[mathescape=true]
let helper xs:$L^{2,2}(\mathbb{Z})$ a b c =					       (* 1 *)
	match xs with
	| [] -> 					       (* 1 *)
		tick 1; [(a,b,c)] 			       (* 0 *)
	| hd::(tl:$L^{6,6}(\mathbb{Z})$) ->			 		       (* 3 *)
		(* share tl:$L^{6,6}(\mathbb{Z})$ as $L^{2,2}(\mathbb{Z})$, $L^{2,2}(\mathbb{Z})$, $L^{2,2}(\mathbb{Z})$ *)
		let newA = hd::a in			       (* 3 *)
		let tmp1 = helper tl:$L^{2,2}(\mathbb{Z})$ newA b c in              (* 2 *)
		let newB = hd::b in			       (* 2 *)
		let tmp2 = helper tl:$L^{2,2}(\mathbb{Z})$ a newB c in              (* 1 *)
		let newC = hd::c in			       (* 1 *) 
		let tmp3 = helper tl:$L^{2,2}(\mathbb{Z})$ a b newC in              (* 0 *)
		tmp1 @ tmp2 @ tmp3		               (* 0 *)
		
let ballBins3 xs:$L^{2,2}(\mathbb{Z})$ = 					       (* 1 *)
	helper xs:$L^{2,2}(\mathbb{Z})$ [] [] []				       (* 0 *)
\end{lstlisting}
\vspace{-1.5ex}
 
By paying a unit of resource for each such way using $\mathit tick$, we can use AARA to bound the count. It assigns a type of $L^{2,2}(\mathbb{Z}) \stackrel {1/0} {\rightarrow} L^{0,0}(L^{0,0}(\mathbb{Z})\times L^{0,0}(\mathbb{Z})\times L^{0,0}(\mathbb{Z}))$ to $\mathit ballBins3$, where the superscript tracks $\stirling {n+1} 2$ and $\stirling {n+1} 3$ potential, respectively. Since $2\stirling {n+1} 3 + 2\stirling{n+1} 2 + 1 = 3^n$, this bound is exact.

\vspace{-3ex}
\section{Mixed Potential}
\vspace{-2ex}

It is possible to combine the existing polynomial potential functions with these new exponential potential functions to not only conservatively extend both, but further represent potentials functions with their products. This space represents functions in $\Theta(n^k(b+1)^n)$ for naturals $k,b$, and does so with terms of the form $\binom n k \stirling {n+1}{b+1}$ so that $\phi(n,P) = \sum_{b,k} p_{b,k} \cdot \binom n k \stirling {n+1}{b+1}$. Note that for $k$ or $b$ equal to 0, the potential functions here reduce to the offset Stirling numbers or binomial coefficients, respectively.

The methods used to combine these potential functions here can easily be generalized to combine any two suitable sets.

\vspace{-2ex}
\paragraph{Simple Shift Operation}

It is straightforward to find a linear recurrence for these products by distributing over their linear recurrences.
\vspace{-1.5ex}
\large
\begin{align*}
    \scriptstyle\binom {n+1} {k+1} \stirling {n+2}{b+2} &\scriptstyle= (\binom n {k+1} + \binom n k) ((b+2)\stirling {n+1}{b+2} + \stirling {n+1}{b+1}) 
    \\
    &\scriptstyle= (b+2) \binom n {k+1}\stirling {n+1}{b+2} + (b+2) \binom n k\stirling {n+1}{b+2} + \binom n {k+1}\stirling {n+1}{b+1} + \binom n {k}\stirling {n+1}{b+1}
\end{align*}
\normalsize
As before, this yields a definition for $\delta$ and $\lhd$ with Equation~\ref{eq:shift}. Letting $P$ now be indexed by pairs $b,k$:
$
    \lhd p_{b,k} = (b+1)p_{b,k}  + (b+1)p_{b,k+1} + p_{b+1,k} + p_{b+1,k+1}
$, and $\delta(P) = p_{0,1}+p_{1,0}+p_{1,1}$. Noting that these definitions are linear again yields automatability for finite (2-dimensional) prefixes of the basis.

\vspace{-2ex}
\paragraph{Expressivity}

The product of non-negative, non-decreasing functions is still non-negative and non-decreasing, so products of valid potential functions are still valid. Soundness is preserved by letting $p_0$ be shorthand for the new constant function coefficient $p_{0,0}$ wherever it is used in Theorem \ref{thm: sound}. Moreover, maximality of expressivity is preserved, simply by giving index pairs the ordering relation $(i_1,i_2) \leq (j_1,j_2) \iff i_1 \leq j_1 \wedge i_2 \leq j_2$ and applying Theorem \ref{thm:max}.

\vspace{-2ex}
\paragraph{Example}

Consider bounding the number of Boolean and arithmetic operations in a variation of subset sum: \textit{single-use} subset sum. Here the input may contain duplicate numbers that should be ignored, so as to treat the input as a true set. This is a trivial change to the mathematical problem, but one that real code might have to deal with, depending on the implementation of sets. 

The code can be changed to handle this by removing all later duplicates of each number it reaches, so that later recursive calls will never see the number again. It is easy to create a function $\mathit{remove}$ of type $\mathbb{Z} \times L^{a+1,b,c}(\mathbb{Z}) \stackrel{d/d}\rightarrow L^{a,b,c}(\mathbb{Z})$ to do this for any $a,b,c,d$, where the superscript values represent linear, $\stirling {n+1}{2}$, and $n\stirling {n+1} 2$ potential, respectively. 

One can prove by induction that at most $4*2^n - n - 3$ Boolean or arithmetic operations are required. Although this can be bounded with only exponential functions, the purely exponential potential system cannot reason about the exact (linear) cost associated with $\mathit{remove}$, and overestimates the bound to be in $\theta(3^n)$. This mixed system can provide a better (though still loose) bound of $n2^n + 2*2^n - n -1$, giving a type of $L^{0,2,1}(\mathbb{Z})\times \mathbb{Z} \stackrel{1/0}\rightarrow bool$ to $\mathit subSum1$. After showing this derivation, we will show how to find the exact bound with AARA.

The following is the single-use subset sum code, with comments on each line tracking the amount of available resources on each line. For clarity, we indicate sharing and subtype-weakening locations.
\begin{lstlisting}[mathescape=true]
let subSum1 nums:$L^{0,2,1}(\mathbb{Z})$ target =                                       (* 1 *)
    match nums with
    | [] ->                                                     (* 1 *)
        tick 1; target = 0                                       (* 0 *)
    | hd::(tl:$L^{1,6,2}(\mathbb{Z}))$ ->                                               (* 4 *)
        let otherNums:$L^{0,6,2}(\mathbb{Z})$ = remove hd tl:$L^{1,6,2}(\mathbb{Z})$ in                        (* 4 *)
        tick 1; let newTarg = target - hd in 	                 (* 3 *)
        (* weaken otherNums:$L^{0,6,2}(\mathbb{Z})$ to $L^{0,4,2}(\mathbb{Z})$ *)
        (* share otherNums:$L^{0,4,2}(\mathbb{Z})$ as $L^{0,2,1}(\mathbb{Z})$, $L^{0,2,1}(\mathbb{Z})$ *)
        let withNum = subSum1 otherNums:$L^{0,2,1}(\mathbb{Z})$ newTarg in (* 2 *)
        let without = subSum1 otherNums:$L^{0,2,1}(\mathbb{Z})$ target in               (* 1 *)
        tick 1; withNum || without 			         (* 0 *)
\end{lstlisting}

The difference between initial and ending potential gives the upper bound of $1+2\stirling {n+1} 2 + n*\stirling{n+1}2 = n2^n + 2*2^n - n -1$ Boolean or arithmetic operations.

Note that we use the subtype-weakening rule, throwing away 2 units of $\stirling {n+1} 2$ potential. This indicates why the bound is not tight. Next we show how to improve this bound using potential demotion.

\vspace{-2ex}
\paragraph{Demotion}

There is one special exception to the non-negativity of potential annotations that may be added due to the particular nature of the relation between binomial coefficients and Stirling numbers. It represents the concept of \textit{demoting} exponential potential into polynomial potential. 

The relevant relation is $\stirling {n+1} 2 = 2^n-1 = \sum_{i=1}^\infty \binom n i \geq \sum_{i=1}^k \binom n i$. This allows a unit of $\stirling {n+1} 2$ potential to account for one unit \textit{each} of all non-constant binomial coefficient potentials. We can express this with the following additional subtyping rule. In this rule we interpret the 2-dimensional indexing of the potential annotation as a matrix, and we let $\overrightarrow p$ refer to the vector of potential entries at index coordinates ${0,i}$ for $i \geq 1$.
\vspace{-1ex}
\fontsize{9}{10}
\begin{align*}
\infer[\mathit{Demote}]
    {
        L^P(A) <: L^Q(A)
    }{
        P = R + \begin{bmatrix} 0 & \overrightarrow{p} \\ r & 0 \end{bmatrix}
        &
        Q = R + \begin{bmatrix} 0& \overrightarrow{p}+s* \overrightarrow{1} \\ r-s & 0\end{bmatrix}
    }
\end{align*}
\normalsize
\begin{theorem}
    The demotion rule is sound.
\end{theorem}

\begin{proof}
    We need only show that $C <: D$ implies $\Phi(v:D) \leq \Phi(v:C)$ for unchanged values $v$. The rest of soundness then follows as in Theorem $\ref{thm: sound}$. To do so, it is sufficient to show for $l = [a_1,\dots,a_n]$ we have $\Phi(a:L^{Q}(A)) \leq \Phi(a:L^{P}(A))$. 
    
    Without loss of generality, we need only consider where $R=0$.
   \begin{align*}
       \Phi(l :L^{Q}(A)) =& \phi(n, Q) + {\scriptstyle\sum}_{i=1}^n \Phi(a_i:A)
       \\
       =& (r-s){\scriptstyle\stirling {n+1} 2} + {\scriptstyle\sum}_{i=1}^k (\overrightarrow{p}_{i-1} +s) {\scriptstyle\binom n i} + {\scriptstyle\sum}_{i=1}^n \Phi(a_i:A)
   \end{align*}
   \begin{align*}
       =&{\scriptstyle\sum}_{i=1}^\infty(r-s){\scriptstyle\binom n i} + {\scriptstyle\sum}_{i=1}^k (\overrightarrow{p}_{i-1} +s) {\scriptstyle\binom n i} +{\scriptstyle\sum}_{i=1}^n \Phi(a_i:A)
\\
       \leq& {\scriptstyle\sum}_{i=1}^\infty r{\scriptstyle\binom n i} + {\scriptstyle\sum}_{i=1}^k \overrightarrow{p}_{i-1} {\scriptstyle\binom n i} + {\scriptstyle\sum}_{i=1}^n \Phi(a_i:A)
       \\
       =&  r{\scriptstyle\stirling {n+1} 2} + {\scriptstyle\sum}_{i=1}^k \overrightarrow{p}_{i-1} {\scriptstyle\binom n i }+ {\scriptstyle\sum}_{i=1}^n \Phi(a_i:A)
       \\
       =& \phi(n, P) + {\scriptstyle\sum}_{i=1}^n \Phi(a_i:A) = \Phi(l :L^{P}(A))
   \end{align*}
\end{proof}

As a corollary, this allows us to loosen the constraint that every annotation $P$ contains only non-negative rationals. In particular, it is no longer required that $\forall i. p_{0,i} \geq 0$. Instead, we require that $\forall i. p_{0,i} + p_{1,0} \geq 0$. Each unit of $\stirling {n+1} 2$ potential may ``pay" for one unit of deficit from each polynomial potential function. Because this is still a linear constraint, type inference remains automatable. 

Using $\mathit{Demote}$, tighter bounds can be obtained. Consider the single-use subset sum solution from the previous section. Here it is again below, but this time allowing the linear potential to be paid for by $\stirling {n+1} 2$ potential. AARA can now provide a type of $L^{-1,4,0}(\mathbb{Z})\times \mathbb{Z} \stackrel{1/0}\rightarrow bool$ for $\mathit subSum1$, corresponding to the exact upper bound of $4*2^n - n - 3$ operations. This time $n*\stirling {n+1} 2$ is elided in the annotated potentials, as it is not needed.

\begin{lstlisting}[mathescape=true]
let subSum1 nums:$L^{-1,4}(\mathbb{Z})$ target =                                       (* 1 *)
    match nums with
    | [] ->                                                     (* 1 *)
        tick 1; target = 0  	                                 (* 0 *)
    | hd::(tl:$L^{-1,8}(\mathbb{Z})$) ->                                              (* 4 *)
        let otherNums:$L^{-2,8}(\mathbb{Z})$ = remove hd tl:$L^{-1,8}(\mathbb{Z})$ in                        (* 4 *)
        tick 1; let newTarg = target - hd in 	                 (* 3 *)
        (* share otherNums:$L^{-2,8}(\mathbb{Z})$ as $L^{-1,4}(\mathbb{Z})$, $L^{-1,4}(\mathbb{Z})$ *)
        let withNum = subSum1 otherNums:$L^{-1,4}(\mathbb{Z})$ newTarg in (* 2 *)
        let without = subSum1 otherNums:$L^{-1,4}(\mathbb{Z})$ target in               (* 1 *)
        tick 1; withNum || without  	      		         (* 0 *)
\end{lstlisting}

The difference between initial and ending potential gives the upper bound of $1-n +4\stirling {n+1} 2 = 4*2^n - n - 3$, as desired.

\vspace{-1ex}
\section{Exponentials, Polynomials, and Logarithms}
\label{sec:exppoly}
\vspace{-1ex}

The addition of exponential potential also allows for the inference of previously nonderivable polynomial-resource types for certain programs. One such way this can happen is by compacting the potential of a list into a new list logarithmic in size to the first. Performing exponential-cost operations, such as $\mathit{subsetSum}$, on a list of logarithmic size only has linear cost in total. 

In the code below, $\mathit{log}$ takes a list $x$ of length $n$ and returns a list of length roughly $log_2 (n)$. If $x$ begins with one unit of linear potential, the type system assigns the output of $\mathit{log}$ one unit of base-2 exponential ($2^n-1$) potential. We show in the code below with types of the form $L^{a,b}$, where $a$ is the linear potential, and $b$ is the base-2 exponential potential. This lets us find that $\mathit{half}$ can have type $L^{1,0}(\mathbb{Z}) \stackrel {0/0} \rightarrow L^{2,0}(\mathbb{Z})$ and $\mathit log$ has type $L^{1,0}(\mathbb{Z}) \stackrel {0/0} \rightarrow L^{0,1}(\mathbb{Z})$. The typing of $\mathit log$ shows the conversion from linear to exponential potential.

\begin{lstlisting}[mathescape=true]
let half x$:L^{1,0}(\mathbb{Z})$ =                                       (* 0 *)
    match x with 
    | [] ->                                       (* 0 *)
        []$:L^{2,0}(\mathbb{Z})$                                         (* 0 *)
    | hd::(tl$:L^{1,0}(\mathbb{Z})$) ->                                 (* 1 *)
        match tl with
        | [] ->                                   (* 1 *)
            []$:L^{2,0}(\mathbb{Z})$                                     (* 1 *)
        | hd2::(tl2$:L^{1,0}(\mathbb{Z})$) ->                           (* 2 *)
	    let halfTail$:L^{2,0}(\mathbb{Z})$  = half tl2 in            (* 2 *)
	    (hd::halfTail)$:L^{2,0}(\mathbb{Z})$           	 	   (* 0 *)
	            
let log x$:L^{1,0}(\mathbb{Z})$ =                                        (* 0 *)
    match x with
    | [] ->                                       (* 0 *)
        []$:L^{0,1}(\mathbb{Z})$                                         (* 0 *)
    | hd::(tl$:L^{1,0}(\mathbb{Z})$) ->                                 (* 1 *)
        let halfTail$:L^{2,0}(\mathbb{Z})$ = half tl in                  (* 1 *)
        let subSoln$:L^{0,2}(\mathbb{Z})$ = log halfTail in              (* 1 *)
        (hd::subSoln)$:L^{0,1}(\mathbb{Z})$                   	   (* 0 *)
\end{lstlisting}

Typing $\mathit log$ above requires resource-polymorphic recursion. However, this can be justified by noting that the above can be thought of to show $\mathit{half}$ has type $L^{a,0}(\mathbb{Z}) \stackrel {0/0} \rightarrow L^{2a,0}(\mathbb{Z})$ and $\mathit log$ has type $L^{a,0}(\mathbb{Z}) \stackrel {0/0} \rightarrow L^{0,a}(\mathbb{Z})$ for any $a \geq 0$. 

Coincidentally, $\mathit log$ conversion of linear to exponential potential certifies that the output list's size can be bounded by a logarithm of the input's size. Nonetheless, logarithmic \textit{potential} is not directly compatible with the approach this work takes. Sublinear functions have negative second derivatives, and this yields negative annotation entries under $\lhd$ applications. This may not be insurmountable, as the demotion rule showed here, but new ideas are needed overall. Logarithmic potential has been explored in \cite{moserlogs}, though the approach there departs from the automatable AARA framework of linear constraint solving. 

\vspace{-3ex}
\section{Conclusion and Future Work}
\vspace{-2ex}

Using Stirling numbers of the second kind allows for the automated inference of exponential resource usages via Automatic Amortized Resource Analysis. This may be combined with the existing polynomial system, allowing mixtures of polynomial and exponential functions to be inferred. Under this system, more kinds of programs can now be automatically analyzed, in particular those making use of multiple recursive calls, or logarithmically-sized lists. Finally, the framework put in place to accomplish this separates the concerns of the type system and potential functions, paving the way to allow modular addition of different potential functions. Future work could extend the work here to cover additional language features supported in polynomial AARA literature, like trees \cite{Hoffmann11}. 

\bibliography{publications,lit,otherSources}

\ifshort{

%%%%% To display Open Access text and logo, Please add below text and copy the cc_by_4-0.eps in the manuscript package %%%

\vfill

{\small\medskip\noindent{\bf Open Access} This chapter is licensed under the terms of the Creative Commons\break Attribution 4.0 International License (\url{http://creativecommons.org/licenses/by/4.0/}), which permits use, sharing, adaptation, distribution and reproduction in any medium or format, as long as you give appropriate credit to the original author(s) and the source, provide a link to the Creative Commons license and indicate if changes were made.}

{\small \spaceskip .28em plus .1em minus .1em The images or other third party material in this chapter are included in the chapter's Creative Commons license, unless indicated otherwise in a credit line to the material.~If material is not included in the chapter's Creative Commons license and your intended\break use is not permitted by statutory regulation or exceeds the permitted use, you will need to obtain permission directly from the copyright holder.}

\medskip\noindent\includegraphics{cc_by_4-0.eps}

}{}

\end{document}

%%% Local Variables:
%%% mode: latex
%%% TeX-master: t
%%% End: